\documentclass{article}
\usepackage[T2A]{fontenc}
\usepackage[utf8]{inputenc}
\usepackage[english]{babel}
\usepackage[tbtags]{amsmath}
\usepackage{amsfonts,amssymb,mathrsfs,amscd,comment,amsthm}
\usepackage[noend]{algorithm2e}
\usepackage{natbib}
\usepackage{wasysym}

\newtheorem{theorem}{Theorem}

\newtheorem{lemma}{Lemma}

\DeclareMathOperator{\diag}{diag}

\begin{document}
\title{\bf\large\MakeUppercase{Exponential convergence rate\\for Iterative Markovian Fitting}}
\author{Sokolov K.O.$^{1,2,3}$, Korotin A.A.$^{1,4,}$\footnote{\textbf{Acknowledgement.} The work was supported by the grant for research centers in the field of AI provided by the Ministry of Economic Development of the Russian Federation in accordance with the agreement 000000C313925P4F0002 and the agreement with Skoltech №139-10-2025-033. \textbf{1.} Skolkovo Institute of Science and Technology. \textbf{2.} “Vega” Institute. \textbf{3.} Lomonosov Moscow State University. \textbf{4.} Artificial Intelligence Research Institute.}}%
\date{}

\maketitle

\vspace{-7mm}We study the problem of constructing the Schr\"odinger bridge \cite[\wasyparagraph 2]{CSBM} in discrete time ${0\!=\!t_0\!<\!t_1\!<\!\dots t_N\!<\!t_{N+1}\!=\!1}$ on the set $\mathcal{X}\!=\!\{1,2,\dots,|\mathcal{X}|\}$. Two distributions $\mu,\nu\in\mathcal{P}(\mathcal{X})$ with everywhere positive density are given. Let $q\in \mathcal{P}(\mathcal{X}^{N+2})$ denote a distribution on $\mathcal{X}^{N+2}$. It may be associated with an $\mathcal{X}$-valued process over times $\{t_0,\dots,t_{N+1}\}$. We abbreviate $x_{t_{n}:t_{n+j}}:=(x_{t_n},\dots,x_{t_{n+j}})$ for consecutive indices. We denote the marginal density (probability) of a trajectory segment by $q(x_{t_{n}:t_{n+j}})$, and use $q(x_{t_i}|x_{t_j})$ for conditional densities. For a \text{Markov} process $q(x_{0:1})\!=\!\Pi_{n=0}^N q(x_{t_{n+1}}|x_{t_n})q(x_0)\!>\!0$ one seeks \vspace{-1.5mm} 
\begin{equation}\label{Schrodinger}
    \min_{p\in\Pi_{N+2}(\mu,\nu)} KL(p\|q),
\end{equation}

\vspace{-3mm}\noindent where $KL(p\|q):=\sum_{\omega\in\Omega}\ln\left(\frac{p(\omega)}{q(\omega)}\right)p(\omega)$ is the Kullback–Leibler divergence, $\Pi_{N+2}(\mu,\nu)\subset \mathcal{P}(\mathcal{X}^{N+2})$ is the set of processes with fixed marginals $\mu$ and $\nu$ at the initial and final times ($t_0=0$, $t_{N+1}=1$). The minimum $p^{*}$ exists and is unique.

Recently the IMF algorithm \cite{DSBM} was proposed to solve problem (\ref{Schrodinger}), which consists of successive transformations interpreted as projections onto the sets of \textbf{Markov} and $q$-\textbf{reciprocal} processes (see \cite[\wasyparagraph 2.5]{CSBM}): 
\begin{equation}
    p_{k+\frac{1}{2}}=\Pi_{n=0}^N p_{k}(x_{t_{n+1}}|x_{t_n})p_{k}(x_0), \quad 
    p_{k+1}= q(x_{t_1:t_N}|x_0, x_1)p_{k+\frac{1}{2}}(x_0,x_1),
    \label{markov-reciprocal-proj}
\end{equation}
where $p_0\in \Pi_{N+2}(\mu,\nu)$ takes the form $p_0:=q(x_{t_1:t_N}|x_0,x_1)\eta(x_0,x_1)$ with $\eta\in \Pi(\mu,\nu)$.

Although it is known that $KL(p_k\|p^*)\rightarrow 0$ as $k\rightarrow\infty$ \cite[Corollary 3.2]{CSBM}, the rate of this convergence remains an open question. Here we \textbf{for the first time prove exponential convergence of IMF}. We rely on convergence analysis of iterations \cite{ConvBlockDescent} minimizing a strongly convex function with a Lipschitz gradient. Denote $\varepsilon_q\!:=\!\min_{ \mathcal{X}^{N+2}} q(x_{t_{1}:t_{N}}|x_0, x_1)$, $\varepsilon_\mu\!:=\!\min_{x_0\mathcal{\in X}}\mu(x_0)$, $\varepsilon_\nu\!:=\!\min_{x_1\mathcal{\in X}}\nu(x_1)$ and $m:=\varepsilon_q^{N+2}\varepsilon_\mu\varepsilon_\nu$.
\vspace{-3.5mm}\begin{theorem}  
For $k=1,2,\dots,$ it holds that 
$
KL(p_k\|p^*)\leq \left(1-\frac{m^3}{4}\right)^{k-1}\,KL(p_0\|p^*).
$
\label{main-theorem}
\end{theorem}
\vspace{-3.5mm}\begin{lemma}\label{EquivMarkov}
    The Markov projection satisfies $p_{k+\frac{1}{2}}= \arg\min_{p\in A(p_k)} KL(p\|p^*)$, where
    $
    A(\eta)\!\!:= \bigl\{p\in \mathcal{P}(\mathcal{X}^{N+2}): p(x_{t_n},x_{t_{n+1}})=\eta(x_{t_n},x_{t_{n+1}}),\; \forall n\in\{0,\dots,N\}, \forall x_0,x_1\in\mathcal{X}\bigr\}.
    $
    In turn, the reciprocal projection satisfies $p_{k+1}= \arg\min_{p\in B(p_{k+1/2})} KL(p\|p^*)$, where 
    $
    B(\eta):=\{p\in \mathcal{P}(\mathcal{X}^{N+2}): p(x_0,x_1)=\eta(x_0,x_1), \forall x_0,x_1\in\mathcal{X}\}.
    $
\end{lemma}
\begin{proof}
We recall from \cite[Theorem~3.1]{CSBM} that the solution $p^*$ satisfies the Markov property.
%Recall that the solution $p^*$ is Markov process \cite[Theorem~3.1]{CSBM}. 
We obtain
$
    KL(p\|p^*)=\sum p(x_{0:1})\ln \frac{p(x_{0:1})}{p_{k+\frac{1}{2}}(x_{0:1})} 
    - \sum p(x_{0:1})\ln\frac{p_{k+\frac{1}{2}}(x_{0:1})}{p^*(x_{0:1})}
    =KL(p\|p_{k+\frac{1}{2}})-\sum \Bigl(\sum_{n=0}^N p_k(x_{t_{n}},x_{t_{n+1}})\ln\frac{p_k(x_{t_{n+1}}|x_{t_{n}})}{p^*(x_{t_{n+1}}|x_{t_{n}})}+p(x_{0})\ln\frac{p_k(x_0)}{p^*(x_0)}\Bigr),
$
where the latter term does not depend on $p$, and thus $p_{k+\frac{1}{2}}$ indeed minimizes. For the reciprocal projection, by analogy with \citep[Formula (4)]{CSBM}, we decompose $KL$ over the marginal $p(x_0,x_1)$ and reduce the problem to minimizing $KL\bigl(p(x_{t_1:t_N}|x_0,x_1)\big\|q(x_{t_1:t_N}|x_0,x_1)\bigr)$ for all $(x_0,x_1)$, since $p^*=q(x_{t_1:t_N}|x_0, x_1)p^*(x_0,x_1)$  \cite[Theorem~3.1]{CSBM}.
Because there are no constraints on the conditional distributions $p(x_{t_1:t_N}|x_0,x_1)$, the solution is $q(x_{t_1:t_N}|x_0,x_1)$.
\end{proof}
% \begin{lemma}\label{EquivReciprocal}
   
% \end{lemma}

\begin{lemma}
    There exists a convex set in $\Pi_{N+2}(\mu,\nu)$ containing $p^*$ and the sequence $\{p_{\frac{k}{2}}\}_{k=1}^\infty$. Moreover, on this domain $KL(\cdot\|p^*)$ is 1-strongly convex, and its gradient is $(1/m)$-Lipschitz for $1\geq m>0$ which defined before Theorem \ref{main-theorem}.
\end{lemma}
\vspace{-4mm}\begin{proof}
    Let $\eta=q(x_{t_{1}:t_{N}}|x_0, x_1)\eta(x_0,x_1)$, where $\eta(x_0,x_1)\in\Pi(\mu,\nu)$. Define $\eta_{1/2}(x_{0:1}) \!:=\! \eta(x_0)\prod_{n=0}^N \eta(x_{t_{n+1}}|x_{t_{n}})
    \geq
    \prod_{n=0}^N \eta(x_{t_{n}},x_{t_{n+1}}),
    $
    where the inequality follows from $\eta(x_{t_{n}})\!\leq\! 1$ and Bayes' formula.
    Next,
    $\eta(x_{t_{n}},x_{t_{n+1}})\!=\!\sum_{x_0,x_1}\!q(x_{t_{n}}\!,x_{t_{n+1}}|x_0,x_1)\eta(x_0,x_1)$, which is no less than $\sum_{x_0,x_1}q(x_{t_{1}:t_{N}}|x_0,x_1)\eta(x_0,x_1)\!\geq\!\varepsilon_q\sum_{x_0,x_1}\eta(x_0,x_1)=\varepsilon_q
    $ for $n\!\notin\!\{0,N\}$. Similarly, for $\eta(x_{{0}},x_{t_{1}})$ and $\eta(x_{t_{N}},x_{1})$ we obtain lower bounds $\varepsilon_q\varepsilon_\mu$ and $\varepsilon_q\varepsilon_\nu$ respectively.
    We obtain $\eta_{1/2}(x_{0:1})\geq \varepsilon_q^{N+1}\varepsilon_\mu\varepsilon_\nu=:\delta$, 
    hence $\eta_{1/2}(x_0,x_1)\geq\delta$.
    Let $\eta_1:=q(x_{t_{1}:t_{N}}|x_0, x_1)\eta_{1/2}(x_0,x_1)$, then $\eta_1(x_{0:1})\geq m:=\varepsilon_q \delta$.
    Each element of the sequence  $\{p_{k/2}\}_{k=1}^\infty$, as well as $p^*$, can be represented as $\eta_{1/2}$ or $\eta_{1}$ for some $\eta$. Consequently, $\{p_{k/2}\}_{k=1}^\infty$ and $p^*$ are contained in $\{p\in\Pi_{N+2}(\mu,\nu): p(x_{0:1})\geq m\}$.
    Consider $KL(p\|p^*)$ as a function of $p\in\mathbb{R}_+^{|\mathcal{X}|^{N+2}}$, its Hessian is diagonal with entries $1/p$, yielding the norm bound
    $
    1\leq\|\diag(1/p)\|\leq 1/m.
    $
\end{proof}

\vspace{-2mm}\noindent Let $\xi(x_{0:1})$ be the density of a finite signed measure on $|\mathcal{X}|^{N+2}$. Denote by $-(t_i,t_j)$ the set of indices ${t_0}\!:\!t_{N+1}$ without $t_i$, $t_j$. Define
$
\xi(x_{t_i},x_{t_j}):=\sum_{x_{-(t_i,t_j)}}\xi(x_{0:1}).
$
Set
$
L_A:=\{\xi:\xi(x_{t_{n+1}},x_{t_{n}})\equiv 0,\;\forall n\in\{0,\dots,N\} \},
$
$
L_B:=\{\xi:\xi(x_0,x_{1})\equiv 0\},
$
$
L_C:=\{\xi: \xi(x_0)\equiv 0,\; \xi(x_1)\equiv 0\}
$
— subspaces of signed measures. Note that $A(\eta)=(\eta+L_A)\cap\mathcal{P}(\mathcal{X}^{N+2})$ and $B(\eta)=(\eta+L_B)\cap\mathcal{P}(\mathcal{X}^{N+2})$, where $\eta\in\Pi_{N+2}(\mu,\nu)$.
\vspace{-1mm}\begin{lemma}\label{lemma-orthogonal}
     $\forall \xi\in L_C$: $\|\Pr_{L_A} (\xi)\|^2\!\!+\!\!\|\Pr_{L_B}(\xi)\|^2\!\!\ge\!\!\|\xi\|^2$, where $\Pr$ is the orthogonal projection.
\end{lemma}
\vspace{-3.5mm}\begin{proof} 
    Note that $L_B\subset L_C$, since $\xi(x_0, x_1)\equiv0\Rightarrow\sum_{x_0}\xi(x_0, x_1)\equiv0$. Consider
    $
    L_B^\perp=\{\xi\in L_C: \exists h=h_{\xi}:\mathbb{R}^{\mathcal{X}^2}\to\mathbb{R}\;\text{such that}\;\forall x_{0:1}: \xi(x_{0:1})\equiv h_\xi(x_0,x_1)\},
    $
    i.e.\ the orthogonal complement of $L_B$ in $L_C$. For every $\xi\in L_B^\perp$ we have $\xi\in L_A$, since
    $
    0\equiv\xi(x_0)\equiv\sum_{x_{1}}|\mathcal{X}|^{N}h_\xi(x_0,x_1)
    \;\Rightarrow\;
    0\equiv\sum_{x_{0},x_1}|\mathcal{X}|^{N-2}h_\xi(x_0,x_1)\equiv\xi(x_{t_n},x_{t_{n+1}}).
    $
    Similarly for $n\in\{0,N\}$ use $\xi(x_0)\equiv0$ and $\xi(x_1)\equiv0$. Finally,
    $
    \|\xi\|^2-\|\Pr_{L_B}(\xi)\|^2=\|\Pr_{L_B^\perp}(\xi)\|^2\leq\|\Pr_{L_A}(\xi)\|^2
    \quad\forall \xi \in L_C.
    $
\end{proof}

\vspace{-2mm}\noindent Let $w^*$ be the minimizer of a 1-strongly convex function $f$ with $(1/m)$-Lipschitz gradient on $(p+L)\cap\mathcal{P}(\mathcal{X}^{N+2})$. Then these properties imply
$
\|\nabla_L f(p)\|\geq\|p-w^*\|,
$
$
f(p)-f(w^*)\geq\frac{m}{2}\|\nabla_L f(p)\|^2,
$
where $\nabla_L$ is the projection of the gradient onto $L$.
\vspace{-1.5mm}\begin{lemma}
    $f(p_k)-f(p_{k+1})\geq\frac{m^3}{4}\|\nabla_{L_C} f(p_k)\|^2$, where $f(p):=KL(p\|p^*)$.
\end{lemma}
\vspace{-2.5mm}\begin{proof}
    From the above remark $f(p_{k})-f(p_{k+\frac{1}{2}})\geq\frac{m}{2}\|\nabla_{L_A} f(p_k)\|^2,$ and $f(p_{k+\frac{1}{2}})-f(p_{k+1})\geq\frac{m}{2}\|\nabla_{L_B} f(p_{k+\frac{1}{2}})\|^2.$
    From Lipschitz continuity of the gradient
    $
    \|\nabla_{L_B} f(p_{k+\frac{1}{2}})-\nabla_{L_B} f(p_k)\|^2
    \leq\frac{1}{m^2}\|p_{k+\frac{1}{2}}-p_k\|^2
    \leq\frac{1}{m^2}\|\nabla_{L_A}f(p_k)\|^2.
    $
    Using $\|\alpha\|^2\le2\|\alpha-\beta\|^2+2\|\beta\|^2$ gives
    $
    \|\nabla_{L_B} f(p_k)\|^2
    \leq 2\|\nabla_{L_B} f(p_{k+\frac{1}{2}})-\nabla_{L_B}f(p_k)\|^2
    +2\|\nabla_{L_B} f(p_{k+\frac{1}{2}})\|^2.
    $
    Combining these and noting $m<1$,
    $
    \|\nabla_{L_B} f(p_k)\|^2+\|\nabla_{L_A} f(p_k)\|^2
    \leq\frac{2}{m^2}(\|\nabla_{L_B} f(p_{k+\frac{1}{2}})\|^2+\|\nabla_{L_A} f(p_k)\|^2).
    $
    
    Hence
    $
    f(p_k)-f(p_{k+1})
    =[f(p_k)-f(p_{k+\frac{1}{2}})]+[f(p_{k+\frac{1}{2}})-f(p_{k+1})]
    \geq\frac{m}{2}(\|\nabla_{L_A} f(p_k)\|^2+\|\nabla_{L_B} f(p_{k+\frac{1}{2}})\|^2)
    \geq\frac{m^3}{4}(\|\nabla_{L_B} f(p_k)\|^2+\|\nabla_{L_A} f(p_k)\|^2)
    \geq\frac{m^3}{4}\|\nabla_{L_C} f(p_k)\|^2
    $ by Lemma \ref{lemma-orthogonal}.
\end{proof}
\vspace{-3.5mm}\begin{proof}[Proof of Theorem 1]
    Combining strong convexity with the inequality $\|\nabla_{L_C} f(p)\|\geq\|p-p^*\|$ yields
    $
    \|\nabla_{L_C} f(p)\|^2\geq f(p).
    $
    Therefore,
    $
    f(p_k)-f(p_{k+1})
    \geq\frac{m^3}{4}\|\nabla_{L_C} f(p_k)\|^2
    \geq\frac{m^3}{4}f(p_k)
    $
    for all $k\geq 1$. To complete the proof, we use $f(p_{1})\leq f(p_0)$ \cite[Proposition B.2]{gushchin2024adversarial}.
\end{proof}

\footnotesize
\vspace{-4mm}\bibliographystyle{plain}\vspace{-3mm}
\bibliography{references}

\begin{thebibliography}{1}

\bibitem{ConvBlockDescent}
Amir Beck and Luba Tetruashvili.
\newblock On the convergence of block coordinate descent type methods.
\newblock {\em SIAM Journal on Optimization}, 23(4):2037--2060, 2013.

\bibitem{gushchin2024adversarial}
Nikita Gushchin, Daniil Selikhanovych, Sergei Kholkin, Evgeny Burnaev, and Alexander Korotin.
\newblock Adversarial schr{\"o}dinger bridge matching.
\newblock {\em Advances in Neural Information Processing Systems}, 37:89612--89651, 2024.

\bibitem{CSBM}
Grigoriy Ksenofontov and Alexander Korotin.
\newblock Categorical {Schr\"odinger} bridge matching.
\newblock In {\em Forty-second International Conference on Machine Learning}, 2025.

\bibitem{DSBM}
Yuyang Shi, Valentin De~Bortoli, Andrew Campbell, and Arnaud Doucet.
\newblock Diffusion schr\"{o}dinger bridge matching.
\newblock In {\em Proceedings of the 37th International Conference on Neural Information Processing Systems}, 2023.

\end{thebibliography}

\end{document}